\documentclass[12pt]{article}
\usepackage{algorithm}
\usepackage{myalgorithmic}
\usepackage{amsmath}
\usepackage{amsthm}
\usepackage{fullpage}

\newcommand{\concept}[1]{\textbf{#1}}
\newcommand{\meth}[1]{\textup{#1}}
\newcommand{\ty}[1]{\ensuremath{\textup{\textsf{#1}}}}
\newcommand{\val}[1]{\ensuremath{\textup{#1}}}
\newcommand{\var}[1]{\ensuremath{\textup{\texttt{#1}}}}

\DeclareMathOperator{\loc}{loc}
\DeclareMathOperator{\row}{row}
\DeclareMathOperator{\lstart}{line\ref{line:enq-start}}
\DeclareMathOperator{\lalloc}{line\ref{line:deq-alloc}}
\DeclareMathOperator{\orderpt}{orderpt}
\DeclareMathOperator{\lex}{lex}
\newcommand{\ltlex}{\ensuremath{<_{\lex}}}

\newtheorem{theorem}{Theorem}
\newtheorem{lemma}[theorem]{Lemma}

\title{Two-enqueuer queue in Common2}
\author{David Eisenstat\thanks{No current affiliation. 55 Autumn Street, New Haven, CT 06511, USA. \texttt{eisenstatdavid@gmail.com}}}
\date{}
\begin{document}
\maketitle

\begin{abstract}
The question of whether all shared objects with consensus number 2 
belong to Common2, 
the set of objects that can be implemented in a wait-free manner 
by any type of consensus number 2, 
was first posed by Herlihy. 
In the absence of general results, 
several researchers have obtained implementations 
for restricted-concurrency versions of FIFO queues. 
We present the first Common2 algorithm for a queue 
with two enqueuers and any number of dequeuers. 
\end{abstract}
\section{Introduction} Many concurrent algorithms employ first-in first-out (FIFO) queues, 
making the quality of queue implementations 
by particular synchronization primitives a practical concern. 
In this work, 
we restrict our attention to \concept{wait-free} implementations, 
where processes cannot take infinitely many steps 
without completing one of their operations. 
Wait-freedom is an especially strong fault-tolerance property, 
ensuring that processes make progress 
despite contention and unexpected delays; 
unsurprisingly, 
there are a number of impossibility results 
regarding wait-free implementations. 
Many of these follow from the consensus hierarchy 
of Herlihy~\cite{DBLP:journals/toplas/Herlihy91}, 
who defined the \concept{consensus number} of a data type $\ty{T}$ 
to be the least upper bound on all $n$ 
such that an $n$-process system 
with some collection of objects of type $\ty{T}$ or $\ty{Register}$ 
can implement consensus. 
Since the composition of wait-free simulations is wait-free, 
no type can implement a type with a higher consensus number. 
For example, \ty{Register}, which has consensus number $1$, 
cannot implement \ty{\ty{Queue}}, which has consensus number $2$. 

However, 
the consensus hierarchy does not let us determine 
the structure of the ``can implement'' relation 
for types with the same consensus number. 
Herlihy~\cite{DBLP:journals/toplas/Herlihy91} 
showed that in an $n$-process system, 
any type with consensus number $n' \ge n$ is \concept{universal}, 
that is, it can implement all types. 
He asked whether \ty{Fetch\&Add}, which has consensus number $2$, 
can implement all types with consensus number $2$ 
in systems with three or more processes. 
Several researchers have found implementations for specific types, 
but as of this writing, 
neither a universal implementation nor a counterexample is known. 

Afek, Weisberger, and Weisman showed 
that any type with consensus number $2$ can implement
\ty{Fetch\&Add}~\cite{DBLP:journals/jal/AfekW99,DBLP:conf/podc/AfekWW92} 
and \ty{Swap}~\cite{DBLP:conf/podc/AfekWW92,theses/Weisman}.\footnote{
In turn, 
\ty{Fetch\&Add} can implement all read-modify-write (RMW) types 
with commuting updates, 
and \ty{Swap} can implement all RMW types with overwriting updates. 
}
They defined Common2 to be the set of types 
that can be implemented by any type of consensus number $2$. 
Afek, Gafni, and Morrison~\cite{DBLP:journals/dc/AfekGM07} 
showed that \ty{Stack} is in Common2, 
improving on an implementation for two pushers 
by David, Brodsky and Fich~\cite{DBLP:conf/wdag/DavidBF05}. 
The status of \ty{Queue} remains unknown, however, 
despite the existence of several restricted implementations. 
When all enqueue operations have the same argument, 
\ty{Queue} and \ty{Stack} have the same specification, 
and the one-value \ty{Stack} implementation 
by David, Brodsky, and Fich~\cite{DBLP:conf/wdag/DavidBF05} 
is also a one-value \ty{Queue} implementation. 
Li~\cite{theses/Li} obtained an implementation 
for multiple values and one dequeuer from an algorithm 
by Herlihy and Wing~\cite{DBLP:journals/toplas/HerlihyW90}. 
He extended it to two dequeuers 
via the universal implementation technique 
and conjectured that there is no three-dequeuer implementation. 
David~\cite{DBLP:conf/wdag/David04,theses/David} 
refuted this conjecture by giving 
an implementation for one enqueuer and any number of dequeuers, 
observing, however, that its enqueue operation 
is not amenable to the same technique. 
We describe a variant of David's algorithm 
that admits a two-enqueuer extension, 
leaving open the case of three enqueuers and three dequeuers. 
The known queue implementations are summarized in Table~\ref{table:algos}. 

Given that modern architectures typically offer 
a primitive of consensus number $\infty$, 
our implementation is of mainly theoretical interest, 
though we believe that it contributes to a better understanding 
of the synchronization required to implement \ty{Queue}. 
For this reason, we have not attempted to reduce the space requirements of our algorithms. 

\begin{table}
\centering
\caption{Summary of known wait-free queue implementations from a type of consensus number $2$ in an $n$-process system}
\begin{tabular}{c c c c}
\hline \hline
\textbf{Enqueuers} & \textbf{Dequeuers} & \textbf{Distinct values} & \textbf{References} \\
\hline \hline
$1$ & $n$ & arbitrary & David~\cite{DBLP:conf/wdag/David04,theses/David} \\
$n$ & $2$ & arbitrary & Li~\cite{theses/Li} \\
$n$ & $n$ & $1$ & David, Brodsky, and Fich~\cite{DBLP:conf/wdag/DavidBF05} \\
$2$ & $n$ & arbitrary & this work \\
\hline \hline
\end{tabular}
\label{table:algos}
\end{table}

\section{Model} The setting for this work 
is the standard asynchronous shared-memory model. 
We describe this model only informally; 
the interested reader should consult a formal description 
such as the one by Herlihy~\cite{DBLP:journals/toplas/Herlihy91}. 

A \concept{shared-memory system} consists 
of $n$ sequential \concept{processes} 
and a collection of shared (base) \concept{objects}. 
Processes communicate with other processes 
by performing operations on the objects. 
Each object has a \concept{type}, 
which specifies the sequential behavior 
of the methods that it supports 
as functions from an object state to a return value and a new state. 
Table~\ref{table:types} lists each type used in this paper 
along with its consensus number, 
the methods that it supports, and their defining functions. 
A \concept{schedule} is an arbitrary sequence of processes; 
in the wait-free setting, there are no fairness conditions. 
Each schedule gives rise to an \concept{execution}, 
where starting from some initial state, 
the processes take \concept{steps} according to the schedule. 
When a process takes a step, 
it selects an operation based on the return values of past operations 
and performs it atomically. 

In order to reason about wait-free implementations, 
we augment the base objects with a virtual object 
of the type being implemented. 
Whenever a process attempts to perform an operation on the latter, 
control is transferred to a black-box subroutine, 
which simulates the operation 
by performing finitely many operations on base objects 
and returning a value. 
The correctness property that we consider is 
\concept{linearizability}~\cite{DBLP:journals/toplas/HerlihyW90}. 
In an execution with operations $o_1$ and $o_2$ 
on the virtual object (\concept{virtual operations} hereafter), 
the operation $o_1$ \concept{precedes} the operation $o_2$ 
if $o_1$ returns before $o_2$ is invoked. 
An execution is \concept{linearizable} 
if there exists a total order $\prec$ of virtual operations 
such that first, 
if a virtual operation $o_1$ precedes a virtual operation $o_2$, 
then $o_1 \prec o_2$, and second, 
the return values of the virtual operations 
are consistent with those obtained by performing 
the operations in sequence according to the order $\prec$. 

\begin{table}
\centering
\caption{Types used in this paper}
\begin{tabular}{c c c c}
\hline \hline
\textbf{Type} & \textbf{Consensus number} & \textbf{Method} & \textbf{Defining function} \\
& & & $\textup{Obj.\ State} \rightarrow \textup{Return Val.} \times \textup{Obj.\ State}$ \\
\hline \hline
\ty{Consensus} & $\infty$ & $\meth{decide}(x)$ & $y \mapsto \begin{cases}(x, x) & \textup{if }y = \bot\\(y, y) & \textup{if }y \neq \bot\end{cases}$ \\
\hline
\ty{Fetch\&Add} & $2$ & $\meth{f\&a}(x)$ & $y \mapsto (y, y + x)$ \\
\hline
\ty{Queue} & $2$ & $\meth{deq}()$ & $\langle \rangle \mapsto (\bot, \langle \rangle)$ \\
& & & $\langle x \rangle \circ q' \mapsto (x, q')$ \\
& & $\meth{enq}(x)$ & $q \mapsto (\val{Ok}, q \circ \langle x \rangle)$ \\
\hline
\ty{Register} & $1$ & $\meth{read}()$ & $y \mapsto (y, y)$ \\
& & $\meth{write}(x)$ & $y \mapsto (\val{Ok}, x)$ \\
\hline
\ty{Stack} & $2$ & $\meth{pop}()$ & $\langle \rangle \mapsto (\bot, \langle \rangle)$ \\
& & & $s' \circ \langle x \rangle \mapsto (x, s')$ \\
& & $\meth{push}(x)$ & $s \mapsto (\val{Ok}, s \circ \langle x \rangle)$ \\
\hline
\ty{Swap} & $2$ & $\meth{swap}(x)$ & $y \mapsto (y, x)$ \\
\hline \hline
\end{tabular}

$\langle \cdots \rangle$ denotes a sequence.
$\circ$ denotes concatenation.
$\bot$ is a return value that indicates failure.
$\val{Ok}$ indicates success in the absence of a value to return.
\label{table:types}
\end{table}

\section{Queue implementations} David's~\cite{DBLP:conf/wdag/David04,theses/David} 
and Li's~\cite{theses/Li} implementations 
can be thought of as variations on Algorithm~\ref{algo:sesd}, 
a simple algorithm in which a single enqueuer writes 
the enqueued items in order for consumption by a single dequeuer. 
At the core of both implementations is the idea 
that either the enqueuers or the dequeuers, but not both, 
can access the array out of order. 

In Li's algorithm, 
enqueuers divide up the locations in the array 
with a $\ty{Fetch\&Add}$ object. 
Because an enqueuer may stall in the interval 
between reserving a location and writing it, 
items may be written out of order---an unavoidable consequence 
of not having a primitive able to achieve consensus among enqueuers. 
To cope, the dequeuer searches all reserved locations for an item; 
fortunately, 
it need not consider locations reserved after the dequeue began. 
Since the only operations performed by the dequeuer on shared objects are reads, 
a type of consensus number $n$ allows $n$ dequeuers to simulate a single dequeuer and schedule their dequeue operations on that dequeuer by Herlihy's universal construction. 

David's algorithm takes the opposite approach, 
where the dequeuers divide up the array. 
Unfortunately, 
a dequeuer may reserve a location 
to which the enqueuer has not yet written, 
in which case we say that the dequeuer 
has \concept{overtaken} the enqueuer. 
The simple solutions, 
where the dequeuer either waits for a value 
or just returns $\bot$, 
are not sufficient; 
the result is an algorithm that is not wait-free 
or that loses enqueued items. 

David's solution to this problem 
is for the enqueuer to recognize when it has been overtaken 
and try again in a way that guarantees success. 
The array of items becomes 
a two-dimensional array of $\ty{Swap}$ objects, 
and dequeuers read locations destructively 
by swapping in a value $\top$ distinct from the initial value $\bot$. 
When the enqueuer is overtaken, it swaps out the value $\top$. 
It is in this case that the second dimension is used: 
the enqueuer writes the item to the beginning of the next row 
before informing the dequeuers that this row is now the current one. 
The dequeuers that reserved empty locations in the previous row 
return $\bot$, 
and their operations can be linearized just before the enqueue, 
when the queue is empty. 

There is no straightforward adaptation 
of David's algorithm to two enqueuers, 
because with two enqueuers swapping an item into the same location, 
the second swap may return the item, 
leaving the enqueuer that performs it 
unsure as to whether the other swap returns $\top$ or $\bot$. 
In Algorithm~\ref{algo:semd}, 
we use a different mechanism for detecting 
when the enqueuer has been overtaken. 
Before a dequeuer begins operating on a location $(i, j)$, 
it writes $\val{true}$ to $\var{deqActive}[i, j]$. 
When the enqueuer finishes with a location $(i, j)$, 
it reads $\var{deqActive}[i, j]$. 
If the read returns $\val{true}$, 
the enqueuer assumes that it has been overtaken. 
This conservative assumption is not always correct, 
and without further modifications, some items may be returned twice! 
We add a layer of indirection to address this issue: 
the two-dimensional array contains \concept{indexes} of items, 
and the dequeuers use a $\ty{Fetch\&Add}$ object 
to establish exclusive ownership. 
A dequeuer that fails to win an item must retry; 
by retrying in the same row, 
it turns out that at most two retries are necessary. 

Unlike David's algorithm, Algorithm~\ref{algo:semd} is amenable to an extension of Li's trick. 
We present the modified enqueue method following the proof of correctness for one enqueuer. 

\begin{algorithm}
\caption{Single-enqueuer single-dequeuer queue (folklore)}
\begin{algorithmic}[1]
\STATE $\var{head} : \ty{integer}$ \COMMENT {enqueuer-local; initially $0$}
\STATE $\var{item} : \arrayof{0..}{\ty{item}}$ \COMMENT {initially $\bot$}
\STATE $\var{tail} : \ty{integer}$ \COMMENT {dequeuer-local; initially $0$}
\bigskip
\METH {$\meth{enq}(\var{x} : \ty{item})$}
    \STATE $\var{item}[\var{head}] := \var{x}$
    \STATE $\var{head} := \var{head} + 1$
\ENDMETH
\bigskip
\METH {$\meth{deq}() : \ty{item}$}
    \STATE $\var{x} := \var{item}[\var{tail}]$
    \IF {$\var{x} \neq \bot$}
        \STATE $\var{tail} := \var{tail} + 1$
    \ENDIF
    \RETURN {$\var{x}$}
\ENDMETH
\end{algorithmic}
\label{algo:sesd}
\end{algorithm}

\begin{algorithm}
\caption{Single-enqueuer multiple-dequeuer queue}
\begin{algorithmic}[1]
\STATE $\var{deqActive} : \arrayof{0.., 0..}{\ty{boolean}}$ \COMMENT {initially $\val{false}$}
\STATE $\var{enqCount} : \ty{integer}$ \COMMENT {enqueuer-local; initially $0$}
\STATE $\var{head} : \ty{integer}$ \COMMENT {enqueuer-local; initially $0$}
\STATE $\var{item} : \arrayof{1..}{\ty{item}}$
\STATE $\var{itemIndex} : \arrayof{0.., 0..}{\ty{integer}}$ \COMMENT {initially $0$}
\STATE $\var{itemTaken} : \arrayof{1..}{\ty{Fetch\&Add}}$ \COMMENT {initially $0$}
\STATE $\var{row} : \ty{integer}$ \COMMENT {initially $0$}
\STATE $\var{tail} : \arrayof{0..}{\ty{Fetch\&Add}}$ \COMMENT {accessed only by dequeuers; initially $0$}
\bigskip
\METH {$\meth{enq}(\var{x} : \ty{item})$}
    \STATE $\var{enqCount} := \var{enqCount} + 1$ \label{line:enq-start}
    \STATE $\var{item}[\var{enqCount}] := \var{x}$
    \STATE $\var{itemIndex}[\var{row}, \var{head}] := \var{enqCount}$
    \IF {$\var{deqActive}[\var{row}, \var{head}]$}
        \STATE $\var{itemIndex}[\var{row} + 1, 0] := \var{enqCount}$
        \STATE $\var{head} := 1$
        \STATE $\var{row} := \var{row} + 1$
    \ELSE
        \STATE $\var{head} := \var{head} + 1$
    \ENDIF
\ENDMETH
\bigskip
\METH {$\meth{deq}() : \ty{item}$}
    \STATE $\var{i} := \var{row}$
    \LOOP
        \STATE $\var{j} := \var{tail}[\var{i}].\meth{f\&a}(1)$ \label{line:deq-alloc}
        \STATE $\var{deqActive}[\var{i}, \var{j}] := \val{true}$
        \STATE $\var{k} := \var{itemIndex}[\var{i}, \var{j}]$
        \IF {$\var{k} = 0$}
            \RETURN {$\bot$}
        \ELSIF {$\var{itemTaken}[\var{k}].\meth{f\&a}(1) = 0$}
            \RETURN {$\var{item}[\var{k}]$}
        \ENDIF
    \ENDLOOP
\ENDMETH
\end{algorithmic}
\label{algo:semd}
\end{algorithm}

\section{Proof of correctness} The main result in this section is the following theorem, 
which we establish by a sequence of lemmas. 

\begin{theorem}
\label{theorem:correct}
Algorithm~\ref{algo:semd} 
is a wait-free linearizable implementation 
of the type $\ty{Queue}$ for one enqueuer and any number of dequeuers 
from the types $\ty{Fetch\&Add}$ and $\ty{Register}$. 
\end{theorem}

The following lemma implies (bounded) wait-freedom. 

\begin{lemma}
\label{lemma:waitfree}
There is a constant $U$ such that in all executions, 
$\meth{enq}$ and $\meth{deq}$ operations complete in $U$ steps or less. 
\end{lemma}
\begin{proof}
For the $\meth{enq}$ method, which has no loops, this is clear. 
The $\meth{deq}$ method has one loop, 
but upon further examination, 
we find that in the worst case, 
the loop body executes in its entirety at most twice.
If a dequeuer executes the loop body without returning, 
the local variable $\var{k}$ is nonzero, 
and $\var{itemTaken}[\var{k}].\meth{f\&a}(1)$ returns a nonzero value. 
Another dequeuer, then, must set $\var{k}$ to the same value 
and perform $\var{itemTaken}[\var{k}].\meth{f\&a}(1)$ first. 
Both dequeuers read the value of $\var{k}$ 
from locations in the array $\var{itemIndex}$, 
and since each location is accessed by at most one dequeuer, 
this value is written to two different locations. 
Any value written to two locations in the array $\var{itemIndex}$ 
is the largest written to one row 
and the smallest written to the next, 
so it is impossible for a $\meth{deq}$ operation, 
which reads values from only one row, 
to read more than two such values. 
\end{proof}

More difficult is showing that Algorithm~\ref{algo:semd} is linearizable. 
Any execution that is not linearizable 
has a finite prefix that is also not linearizable, 
that is, linearizability is a safety property. 
Moreover, by wait-freedom, 
any finite execution has a finite continuation 
in which processes finish their current queue operations 
without starting new ones. 
If the longer execution is linearizable, 
then so is its prefix, by the same order of operations. 
It thus suffices to show that any finite execution 
where all operations finish is linearizable. 

Fix a particular finite execution where, without loss of generality, 
all operations finish and all enqueued items are distinct. 
We construct a linearization order $\prec$ as follows. 
An $\meth{enq}$ operation $e$ \concept{matches} 
a $\meth{deq}$ operation $d$ 
if $e$ enqueues the item that $d$ dequeues. 
For $\meth{deq}$ operations $d$, 
let $\loc(d)$ be the last location $(i, j)$ 
of $\var{itemIndex}$ read by $d$. 
For $\meth{enq}$ operations $e$ 
that write exactly one location $(i, j)$ 
in the array $\var{itemIndex}$, let $\loc(e) = (i, j)$. 
For $\meth{enq}$ operations $e$ 
that write two locations $(i, j)$ and $(i + 1, 0)$, 
there is a unique $\meth{deq}$ operation $d$ 
that writes $\var{deqActive}[i, j]$. 
Let $\loc(e) = (i, j)$ if $e$ matches $d$ 
and let $\loc(e) = (i + 1, 0)$ otherwise. 
For operations $o$, let $\row(o)$ be the first coordinate of $\loc(o)$. 

\begin{lemma}
\label{lemma:unique-match}
No operation matches more than one other operation. 
\end{lemma}
\begin{proof}
By assumption, 
no item is enqueued more than once, 
so no item is written to two locations in the array $\var{item}$. 
In order to return an item $\var{item}[k]$, 
a $\meth{deq}$ operation $d$ must be the first 
to access $\var{itemTaken}[k]$, 
ensuring that $d$ and the $\meth{enq}$ operation 
that writes $\var{item}[k]$ are uniquely matched. 
\end{proof}

\begin{lemma}
\label{lemma:match-loc}
If an $\meth{enq}$ operation $e$ 
matches a $\meth{deq}$ operation $d$, 
then $d$ does not precede $e$ and $\loc(e) = \loc(d)$. 
\end{lemma}
\begin{proof}
The operation $d$ reads the index of the enqueued item 
from the same location to which $e$ writes that index. 
Consequently, $d$ cannot precede $e$, 
and $\loc(e) = \loc(d)$ by definition. 
\end{proof}

For $\meth{enq}$ operations $e$, 
let $\orderpt(e) = \lstart(e)$ be the time 
at which $e$ executes line~\ref{line:enq-start}, 
where the \concept{time} at which a step is taken 
is the total number of steps that are taken before it. 
For $\meth{deq}$ operations $d$, 
let $\lalloc(d)$ be the latest time 
at which $d$ executes line~\ref{line:deq-alloc}. 
If $d$ matches an $\meth{enq}$ operation $e$, 
let $\orderpt(d) = \max(\lalloc(d), \lstart(e) + \frac{1}{2})$; 
otherwise, let $\orderpt(d) = \lalloc(d)$. 
For operations $o_1$ and $o_2$, 
write $o_1 \prec o_2$ 
if $(\row(o_1), \orderpt(o_1)) \ltlex (\row(o_2), \orderpt(o_2))$, 
where the symbol $\ltlex$ denotes lexicographic order. 

\begin{lemma}
\label{lemma:total-order}
The relation $\prec$ is a total order. 
\end{lemma}
\begin{proof}
It suffices to show that the function $\orderpt$ is one-to-one. 
For operations $o$, 
either $o$ is unique in taking a step at time $\orderpt(o)$, 
or $o$ is a $\meth{deq}$ operation 
that matches an $\meth{enq}$ operation $e$ 
and $\orderpt(o) = \lstart(e) + \frac{1}{2}$. 
In the latter case, 
no operation $o' \neq o$ satisfies $\orderpt(o') = \orderpt(o)$, 
since by Lemma~\ref{lemma:unique-match}, 
the only operation that matches $e$ is $o$. 
\end{proof}

\begin{lemma}
\label{lemma:precedes}
If $o_1$ and $o_2$ are operations such that $o_1$ precedes $o_2$, 
then $o_1 \prec o_2$. 
\end{lemma}
\begin{proof}
Assume that $o_1 \not\prec o_2$. 
If $\row(o_1) > \row(o_2)$, then $o_1$ does not precede $o_2$, 
since the value of $\var{row}$ is nondecreasing. 
Otherwise, $\row(o_1) = \row(o_2)$ and $\orderpt(o_1) \ge \orderpt(o_2)$. 
For all operations $o$, 
the time $\orderpt(o)$ occurs during $o$, 
since either $o$ takes a step at that time, 
or $o$ is a $\meth{deq}$ operation 
that matches an $\meth{enq}$ operation $e$, 
in which case $o$ ends after time $\orderpt(e) = \lstart(e)$ 
by Lemma~\ref{lemma:match-loc}. 
It follows that $o_1$ does not precede $o_2$. 
\end{proof}

\begin{lemma}
\label{lemma:loc-order}
If $e_1$ and $e_2$ are $\meth{enq}$ operations, 
then $e_1 \prec e_2$ if and only if $\loc(e_1) \ltlex \loc(e_2)$. 
If $d_1$ and $d_2$ are $\meth{deq}$ operations, 
then $(\row(d_1), \lalloc(d_1)) \ltlex (\row(d_2), \lalloc(d_2))$ 
if and only if $\loc(d_1) \ltlex \loc(d_2)$. 
\end{lemma}
\begin{proof}
There is only one enqueuer, 
and the pair $(\var{row}, \var{head})$ 
increases lexicographically with each $\meth{enq}$ operation. 
Line~\ref{line:deq-alloc} is the invocation of $\meth{f\&a}$ 
where $\loc(d)$ is obtained. 
\end{proof}

\begin{lemma}
\label{lemma:complete}
If $d$ is a $\meth{deq}$ operation, 
then for all $\meth{enq}$ operations $e'$ 
with $\loc(e') \ltlex \loc(d)$, 
there exists a $\meth{deq}$ operation $d'$ 
that matches $e'$. 
\end{lemma}
\begin{proof}
Fix an $\meth{enq}$ operation $e'$ with $\loc(e') \ltlex \loc(d)$. 
It suffices to show that some process reads the index written by $e'$, 
since it follows that some $\meth{deq}$ operation matches $e'$. 
If $e'$ writes exactly one location $(i, j)$ 
in the array $\var{itemIndex}$, 
then no dequeuer reads that location beforehand, 
as otherwise the enqueuer would read $\val{true}$ 
from $\var{deqActive}[i, j]$. 
Nevertheless, some $\meth{deq}$ operation does perform the read. 
In each row, 
the set of locations read by dequeuers is a prefix of the row, 
and some dequeuer reads a location to the right of $e'$. 
If $i < \row(d)$, 
a suitable witness is the $\meth{deq}$ operation 
that causes the variable $\var{row}$ to be incremented; 
if $i = \row(d)$, a suitable witness is $d$ itself. 
When $e'$ writes two locations of the array $\var{itemIndex}$, 
the second write necessarily precedes any corresponding read, 
since it is performed before the enqueuer increments $\var{row}$. 
The remaining arguments parallel the one-write case, 
with one complication: 
it may be the case that $\loc(d)$ is between 
the locations of the first and second write. 
In this case, $\loc(e') < \loc(d)$ if and only if 
the $\meth{deq}$ operation that triggered the second write matches $e$. 
\end{proof}

\begin{lemma}
\label{lemma:linearize}
The order $\prec$ is a valid linearization order. 
\end{lemma}
\begin{proof}
Given Lemmas~\ref{lemma:total-order}~and~\ref{lemma:precedes}, 
the only property remaining to be established 
is that the return values are consistent 
with the sequential execution determined by the order $\prec$. 
We prove this by induction on the number of operations. 

Specifically, 
the inductive hypothesis is that through $m$ operations, 
all return values are correct, 
and the contents of the queue 
are the items that have been enqueued but not dequeued, 
in the order in which they were enqueued. 
The basis $m = 0$ is trivial. 
Assuming the inductive hypothesis for $m$, 
if the next operation is an $\meth{enq}$ operation, 
the inductive hypothesis holds for $m + 1$, 
since by Lemma~\ref{lemma:match-loc}, 
$\meth{enq}$ operations are not preceded 
by matching $\meth{deq}$ operations. 
If the next operation is a $\meth{deq}$ operation $d$, 
then by Lemma~\ref{lemma:complete}, 
every $\meth{enq}$ operation $e'$ with $\loc(e') \ltlex \loc(d)$ 
has a matching $\meth{deq}$ operation $d'$. 
Each such $d'$ satisfies $\lalloc(d') < \lalloc(d)$ 
by Lemma~\ref{lemma:loc-order}. 
If $d$ returns $\bot$, 
then by the definition of $\prec$, 
it is the case that $e' \prec d$ if and only if $d' \prec d$, 
so the queue is empty and remains empty. 
If $d$ matches an $\meth{enq}$ operation $e$, 
then $e$ is the first $\meth{enq}$ operation not yet matched, 
by a similar argument. 
\end{proof}

We can now prove Theorem~\ref{theorem:correct}. 

\begin{proof}[Proof of Theorem~\ref{theorem:correct}]
Algorithm~\ref{algo:semd} is wait-free by Lemma~\ref{lemma:waitfree} 
and is a linearizable implementation of \ty{Queue} by Lemma~\ref{lemma:linearize}. 
\end{proof}

\begin{theorem}
Algorithm~\ref{algo:semd} can be implemented by any type of consensus number $2$. 
\end{theorem}
\begin{proof}
By the results of Afek, Weisberger, and Weisman~\cite{DBLP:journals/jal/AfekW99,DBLP:conf/podc/AfekWW92}, any type of consensus number $2$ can implement $\ty{Fetch\&Add}$. 
\end{proof}

\section{The two-enqueuer case} The two-enqueuer adaptation of Algorithm~\ref{algo:semd} is presented as Algorithm~\ref{algo:temd}. 
The main idea is the same as in Li's adaptation, although the details are more complicated: operations by two real processes are scheduled onto one virtual process, which makes progress as long as either real process is active. 
This scheduling is accomplished by an $\ty{Agenda}$ object, with a sequential implementation presented as Algorithm~\ref{algo:agenda}. 
Herlihy's universal construction gives a two-process implementation from any type of consensus number $2$. 

Once an enqueuer schedules an enqueue operation $e$, it performs the steps that the enqueuer of Algorithm~\ref{algo:semd} would have up to the point where $e$ is complete. 
Only finitely many enqueue operations precede $e$, so this takes only finitely many steps. 
Exactly once per operation, the enqueue method reads a shared register. 
To ensure that both enqueuers continue to simulate the same trajectory, they reach consensus on the value of that read. 

\begin{theorem}
\label{theorem:correct2}
Algorithm~\ref{algo:temd} is a wait-free linearizable implementation of the type $\ty{Queue}$ for two enqueuers and any number of dequeuers that can be implemented by any type of consensus number $2$. 
\end{theorem}
\begin{proof}[Proof sketch]
The new enqueue method is clearly wait-free. 
Wait-freedom of the new dequeue method and linearizability follow from the fact that each execution of Algorithm~\ref{algo:temd} begets an execution of Algorithm~\ref{algo:semd} that has the same collection of enqueue operations, is indistinguishable to the dequeuers, and in which the ``real'' enqueue operations are active on a super-interval of the corresponding ``virtual'' enqueue operations. 
The real enqueuers both take essentially the same steps as the virtual enqueuer, and the virtual enqueuer is deemed to have taken a particular step when it is first taken by a real enqueuer. 
The construction is made possible by the fact that all of the steps that involve objects shared with the dequeuers are idempotent. 
There are several categories: reads; enqueuer writes to registers that are written exactly once; and writes to $\var{row}$. 
The latter are idempotent because the values written to $\var{row}$ increase over time and the dequeuers use only $\max(\var{row})$.
\end{proof}

\begin{algorithm}
\caption{Agenda object (sequential version)}
\begin{algorithmic}[1]
\STATE $\var{item} : \arrayof{1..}{\ty{item}}$
\STATE $\var{tail} : \ty{integer}$ \COMMENT {initially $0$}
\bigskip
\METH {$\meth{append}(\var{x} : \ty{item}) : \ty{integer}$}
    \STATE $\var{tail} := \var{tail} + 1$
    \STATE $\var{item}[\var{tail}] := \var{x}$
    \RETURN {$\var{tail}$}
\ENDMETH
\bigskip
\METH {$\meth{get}(\var{k} : \ty{integer}) : \ty{item}$}
    \RETURN {$\var{item}[\var{k}]$}
\ENDMETH
\end{algorithmic}
\label{algo:agenda}
\end{algorithm}

\begin{algorithm}
\caption{Two-enqueuer multiple-dequeuer queue}
\begin{algorithmic}[1]
\STATE $\var{agenda} : \ty{Agenda}$ \COMMENT {enqueuer-local; initially empty}
\STATE $\var{deqActive} : \arrayof{0.., 0..}{\ty{boolean}}$ \COMMENT {initially $\val{false}$}
\STATE $\var{deqActiveRead} : \arrayof{0.., 0..}{\ty{Consensus}}$ \COMMENT {enqueuer-local; initially $\bot$}
\STATE $\var{enqCount} : \arrayof{0..1}{\ty{integer}}$ \COMMENT {enqueuer-local; initially $0$}
\STATE $\var{head} : \arrayof{0..1}{\ty{integer}}$ \COMMENT {enqueuer-local; initially $0$}
\STATE $\var{item} : \arrayof{1..}{\ty{item}}$
\STATE $\var{itemIndex} : \arrayof{0.., 0..}{\ty{integer}}$ \COMMENT {initially $0$}
\STATE $\var{itemTaken} : \arrayof{1..}{\ty{Fetch\&Add}}$ \COMMENT {initially $0$}
\STATE $\var{row} : \arrayof{0..1}{\ty{integer}}$ \COMMENT {initially $0$}
\STATE $\var{tail} : \arrayof{0..}{\ty{Fetch\&Add}}$ \COMMENT {accessed only by dequeuers; initially $0$}
\bigskip
\METH {$\meth{enq}(\var{x} : \ty{item})$}
    \STATE $\var{k} := \var{agenda}.\meth{append}(\var{x})$ \COMMENT {returns the index of $\var{x}$ in the agenda}
    \WHILE {$\var{enqCount}[\meth{id}] < \var{k}$}
        \STATE $\var{enqCount}[\meth{id}] := \var{enqCount}[\meth{id}] + 1$
        \STATE $\var{item}[\var{enqCount}[\meth{id}]] := \var{agenda}.\meth{get}(\var{enqCount}[\meth{id}])$
        \STATE $\var{itemIndex}[\var{row}[\meth{id}], \var{head}[\meth{id}]] := \var{enqCount}[\meth{id}]$
        \STATE $\var{b} := \var{deqActive}[\var{row}[\meth{id}], \var{head}[\meth{id}]]$ \label{line:sharedread}
        \IF {$\var{deqActiveRead}[\var{row}[\meth{id}], \var{head}[\meth{id}]].\meth{decide}(\var{b})$}
            \STATE $\var{itemIndex}[\var{row}[\meth{id}] + 1, 0] := \var{enqCount}[\meth{id}]$
            \STATE $\var{head}[\meth{id}] := 1$
            \STATE $\var{row}[\meth{id}] := \var{row}[\meth{id}] + 1$
        \ELSE
            \STATE $\var{head}[\meth{id}] := \var{head}[\meth{id}] + 1$
        \ENDIF
    \ENDWHILE
\ENDMETH
\bigskip
\METH {$\meth{deq}() : \ty{item}$}
    \STATE $\var{i} := \max(\var{row})$
    \LOOP
        \STATE $\var{j} := \var{tail}[\var{i}].\meth{f\&a}(1)$
        \STATE $\var{deqActive}[\var{i}, \var{j}] := \val{true}$
        \STATE $\var{k} := \var{itemIndex}[\var{i}, \var{j}]$
        \IF {$\var{k} = 0$}
            \RETURN {$\bot$}
        \ELSIF {$\var{itemTaken}[\var{k}].\meth{f\&a}(1) = 0$}
            \RETURN {$\var{item}[\var{k}]$}
        \ENDIF
    \ENDLOOP
\ENDMETH
\end{algorithmic}
\label{algo:temd}
\end{algorithm}

\section{Discussion} Algorithm~\ref{algo:temd} also works 
in the \concept{unbounded concurrency} model 
of Gafni, Merritt, and Taubenfeld~\cite{DBLP:conf/podc/GafniMT01}. 
It establishes that two-enqueuer $\ty{Queue}$ 
belongs to the unbounded concurrency version of Common2 
via the $\ty{Fetch\&Add}$ implementation 
due to Afek, Gafni, and Morrison~\cite{DBLP:journals/dc/AfekGM07}. 
Given the unbounded concurrency $\ty{Stack}$ by the same authors 
and a similar adaptation of Li's two-dequeuer $\ty{Queue}$, 
there is currently no set of restrictions 
for which a bounded concurrency algorithm is known 
and an unbounded concurrency algorithm is not. 

Both our algorithm and Li's require that either the enqueuers 
or the dequeuers agree on a total order for the items. 
A general algorithm, if one exists, 
will have to work in the absence of such an agreement, 
though we note that the $\ty{Swap}$ implementation 
of Afek, Weisberger, and Weisman~\cite{DBLP:conf/podc/AfekWW92} 
achieves a similar feat. 
On the other hand, 
the implementation of 
Herlihy and Wing~\cite{DBLP:journals/toplas/HerlihyW90} 
can be modified to be lock-free, 
so any impossibility result will have to distinguish 
lock-free implementations from wait-free ones, 
a property absent from many wait-free impossibility results 
in the literature. 

%\section{Acknowledgments} \input{acks}

\bibliography{paper}
\bibliographystyle{plain}
\end{document}